\documentclass[11pt]{article}
\usepackage{graphicx,latexsym}
\usepackage{amsmath}
\usepackage{amsthm}
\usepackage{amssymb}
\usepackage{colortbl}
\usepackage{array}
\usepackage{tabularx}
\usepackage{stmaryrd}
\newcommand{\avec}{{\bf a}}
\newcommand{\bvec}{{\bf b}}
\newcommand{\xvec}{{\bf x}}
\newcommand{\uvec}{{\bf u}}
\newcommand{\vvec}{{\bf v}}

\newcommand{\yvec}{{\bf y}}
\newcommand{\zvec}{{\bf z}}
\newcommand{\hvec}{{\bf h}}

\newcommand{\lnon}{\overline}

\DeclareMathOperator{\rank}{rank}

\newtheorem{remark}{Remark}
\newtheorem{proposition}{Proposition}

\newtheorem{theorem}{Theorem}

\begin{document}
\title{On the Number of Observation Nodes in Recurrent Neural Networks with Linear Threshold and ReLU Functions}
\author{Liangjie Sun, Wai-Ki Ching, and Tatsuya Akutsu
\thanks{Liangjie~Sun (corresponding author) is with the Institute for Life and Medical
Sciences, Kyoto University, Kyoto 606-8507, Japan (e-mail: ljsun\_seu@126.com).}
\thanks{Wai-Ki Ching is with the Department of Mathematics, The University of
Hong Kong, Pokfulam Road, Hong Kong (e-mail: wching@hku.hk).}
\thanks{Tatsuya Akutsu is with the Bioinformatics Center, Institute for Chemical Research, Kyoto University, Kyoto 611-0011, Japan (e-mail: takutsu@kuicr.kyoto-u.ac.jp).}}

\maketitle

\begin{abstract}
This paper investigates how node update rules and admissible state domains affect the minimum number of observation nodes required for global finite-horizon observability in recurrent neural networks with linear-threshold and ReLU update functions. Over a common binary state domain, we construct a class of $K$-linear-threshold ($K$-LT) networks whose initial states can be uniquely reconstructed from the finite output trajectory of a single observation node. We further establish a dynamical equivalence between binary-valued $K$-ReLU networks and $K$-AND Boolean networks, which transfers existing class-level observation-node bounds to binary-valued $K$-ReLU networks. For nonnegative-valued ReLU networks, the observability problem reduces to the classical linear-system setting whenever the relevant pre-activations remain nonnegative. For general real-valued ReLU networks, we prove that global finite-horizon observability requires at least $n/2$ observation nodes when no restriction is imposed on the number of state variables involved in each node update. This lower bound is tight when $n=2K$, for which we construct a $K$-ReLU network observable from exactly $K$ nodes. These results show that both update rules and state domains fundamentally affect extremal observation requirements: temporal evolution can concentrate finite-state information into a single measured trajectory, whereas activation-induced rank loss creates an intrinsic sensor lower bound in continuous-state ReLU networks.
\end{abstract}


\section{Introduction}\label{sec:introduction}
Observability concerns whether the initial state of a dynamical system can be uniquely determined based on time-varying measurement outputs \cite{Kalman1960, HermannKrener1977, Nijmeijer1982}. It serves as a crucial theoretical foundation for the reliable implementation of state estimation, fault diagnosis, and state-based prediction and control. However, in networked dynamical systems, directly measuring the state of every node is often difficult due to physical, technical, or economic constraints. For instance, in large-scale neural systems, only a very small fraction of neurons may be directly measurable; similarly, in power grids, the cost of installing phasor measurement units at every substation could be prohibitively high. Nevertheless, thanks to the dynamic coupling between nodes, measurements collected from a suitably chosen subset of nodes may still contain sufficient information to uniquely determine the initial state of the entire system \cite{liu13}. This gives rise to two fundamental questions: whether a given set of measurement nodes renders the network observable, and what the minimum number of nodes required to achieve observability is. In this paper, we refer to the nodes selected for measurement as observation nodes.

For finite-dimensional linear time-invariant systems, global state distinguishability is equivalent to the observability matrix satisfying the classical full-rank condition \cite{Kalman1960}. However, for nonlinear systems, several non-equivalent concepts of observability arise. Within the differential-geometric framework proposed by Hermann and Krener, two states are said to be indistinguishable if they produce identical input-output behavior under any admissible input; furthermore, the observability rank condition constitutes a sufficient condition for local weak observability \cite{HermannKrener1977,Sontag1984}. For autonomous discrete-time nonlinear systems, indistinguishability simplifies to the equality of output sequences generated by distinct initial states \cite{Nijmeijer1982}. This paper adopts a global, finite-time observability concept: any pair of distinct initial states must be distinguishable via their output sequences over the same finite time horizon.

ReLU dynamical networks, which can be viewed as recurrent neural networks with ReLU activation and shared parameters across time, form a structured subclass of discrete-time continuous piecewise-affine (PWA) systems. For initial states generating the same activation-pattern sequence over a finite horizon, the state and output trajectories depend affinely on the initial state. Globally, however, the finite-horizon observation map is PWA because the activation-pattern sequence may vary with the initial condition.  As noted in \cite{Bemporad2000,camlibel2006}, observability analysis for hybrid and PWA systems is significantly more complex than for linear systems, since the observability of the overall system generally cannot be inferred solely from the observability of the individual affine subsystems that compose it. Existing research has explored observability issues for PWA and hybrid systems from various perspectives, including mixed-integer tests for finite-step incremental observability \cite{Bemporad2000}, geometric rank conditions for continuous-time switched linear and linear hybrid systems \cite{Vidal2003}, verifiable conditions for trajectory observability in piecewise affine hybrid systems \cite{Collins2004}, and randomized methods for sampled-data PWA systems \cite{AzumaImura2007}. For a given ReLU system defined on a bounded domain, indistinguishability over a finite time horizon can also be tested by unrolling the dynamic process and applying standard mixed-integer encodings for ReLU units or satisfiability-based formulations \cite{Katz2017}. These studies and verification tools address observability or related injectivity tests for a prescribed system, output map, domain, and observation horizon. However, they neither characterize the minimum number of state coordinates that must be measured in the worst case for certain classes of ReLU dynamical networks, nor provide constructive methods to achieve such bounds.

Boolean networks (BNs) offer a complementary finite-state perspective, characterizing network dynamics through binary states and logical update rules \cite{Kauffman1969}. Their finite state space and explicit logical structure make them particularly well-suited for the precise analysis of global observability and the selection of observation nodes. Observability in BNs is closely related to state estimation \cite{wei2025,wang2025,Reveliotis2025}, disturbance decoupling \cite{zhangk2023,yakun2025}, system identification \cite{wang22siam,ji25,li25siam}, and state reconstructibility \cite{Fornasini2019,Yang2020,yang25siam}. Within this framework, a BN is considered observable if any two distinct initial states yield different output sequences within a finite time horizon. Extensive research has been conducted on selecting a minimal set of observation nodes to ensure the observability of a given BN. Existing approaches include graph-theoretic methods \cite{Laschov2013,weiss19}, and algebraic methods based on the matrix semi-tensor product \cite{zhu21,liuyang2022,liy2023,wangli2025,liyalu2025}. These studies primarily focus on optimization problems concerning specific instances: namely, determining the minimum set of observation nodes that guarantees observability, given complete knowledge of the dynamic characteristics of BNs. Recently, the study by \cite{sun24} shifted the research perspective from node selection for specific instances to class-level analysis; using methods such as information theory, it derived general upper and lower bounds on the minimum number of observation nodes required for several classes of BNs. Nevertheless, existing class-level results remain model specific, and a systematic comparison of the minimum numbers of observation nodes required in the best and worst cases for representative finite-state and continuous-state nonlinear network models under a common notion of global finite-horizon observability is still lacking.

Based on the above discussions, this paper focuses on two classes of recurrent neural networks. The first class consists of $K$-linear threshold ($K$-LT) networks, in which each node has a binary state and is updated by a linear threshold function involving $K$ input literals, where a literal refers to either a state variable or its complement. The second class consists of $K$-ReLU networks, in which each node is updated by applying the ReLU activation function to an affine combination of $K$ state variables. Depending on the admissible initial state domains, we further consider binary, nonnegative, and general real-valued $K$-ReLU networks

Under a common notion of global finite-horizon observability, this paper investigates how node update rules and admissible state domains affect the minimum number of observation nodes required for observability. Over the common binary state domain, we compare $K$-LT networks with binary-valued $K$-ReLU networks to isolate the effect of the update rule. Within the $K$-ReLU framework, we compare binary, nonnegative, and general real-valued state domains to determine how the admissible domain changes the realizable activation patterns and the resulting observation requirements. Along these two directions, we conduct a class-level extremal analysis. For $K$-LT networks, existing bounds for general BNs do not accurately characterize the best-case minimum observation requirement and, in particular, do not determine whether global finite-horizon observability can be achieved from a single observation node. We resolve this question by characterizing the extremal observation requirements of $K$-LT networks and providing explicit constructions attaining the corresponding bounds. For $K$-ReLU networks, we derive bounds on the minimum number of observation nodes under different state domains and identify activation patterns and the degrees of freedom in the initial state as key determinants of the observation requirements.

The contributions of this paper are summarized as follows, where $n$ denotes the number of nodes in each network under consideration.
\begin{itemize}
  \item[(i)] Over the binary state domain, we construct a family of $K$-LT networks that is observable from a single node and establish a finite observation horizon within which the initial state is uniquely determined. Since at least one observation node is necessary, the general lower bound and the best-case upper bound coincide at one. Thus, both bounds are tight.

  We further establish an equivalence between binary $K$-ReLU networks and $K$-AND BNs. Consequently, the corresponding general lower bound $[(1-K)+\frac{2^{K}-1}{2^{K}}\log_{2}(2^{K}-1)]n$ and the best-case upper bound $(\frac{2^{K}-K-1}{2^{K}-1})n$ for $K$-AND BNs apply to binary $K$-ReLU networks. Whenever $[(1-K)+\frac{2^{K}-1}{2^{K}}\log_{2}(2^{K}-1)]n>1$, the constructed $K$-LT networks require fewer observation nodes than every binary $K$-ReLU network.
  \item[(ii)] We next examine the effect of the state domain within the ReLU framework. For binary states, the observation-node bounds depend on both $n$ and $K$, as established through the equivalence with $K$-AND BNs. For nonnegative states, under the prescribed positivity assumptions, the ReLU operation does not alter the state update, and the network reduces to a positive affine system. Its observability is then characterized by the classical Kalman rank condition, and a single observation node is sufficient in the best case.

  For general real-valued states, state-dependent ReLU deactivation can reduce the rank of the finite-horizon observation map and make distinct initial states indistinguishable. Without imposing any restriction on the number of state variables on which each node update depends, we establish a general lower bound of $\frac{n}{2}$ on the number of observation nodes required for observability. For $n=2K$, we construct a general $K$-ReLU network and derive an explicit finite-step procedure for reconstructing its initial state from the output sequence of $K$ nodes.
\end{itemize}

The remainder of the paper is organized as follows. Section \ref{sec:problem formulation} introduces the
network models and formulates the minimum-node observability problem.
Section \ref{sec:linear threshold functions} studies the binary-state $K$-LT networks.
Section \ref{sec:relu functions} discusses binary, nonnegative, and general
real-valued $K$-ReLU networks. Section \ref{sec:simulation results} provides
computational validation, and Section \ref{sec:conclusion} concludes the paper.
Additionally, the notation used throughout this paper is summarized in Table \ref{tab:notations}.
\begin{table}[t]
\caption{Some Important Notation}
\label{tab:notations}
\centering
\small
\renewcommand{\arraystretch}{1.05}
\setlength{\tabcolsep}{4pt}

\begin{tabularx}{\columnwidth}{
    >{\raggedright\arraybackslash}p{0.26\columnwidth}
    >{\raggedright\arraybackslash}X
}
\hline
\textbf{Notations} & \textbf{Definitions} \\
\hline
$\mathbb{R}$ & The set of real numbers \\
$\mathbb{R}^{n}$ & The set of $n$-dimensional real column vectors \\
$\mathbb{R}_{\geq 0}^{n}$ & The set of $n$-dimensional nonnegative real column vectors \\
$\mathbb{R}^{n\times m}$ & The set of all $n\times m$ real matrices \\
$\mathbb{R}_{\geq 0}^{n\times m}$ & The set of all $n\times m$ real matrices with nonnegative entries \\
$\mathbb{Z}_{\geq0}$ & The set of nonnegative integers \\
$\mathbb{Z}_{>0}$ & The set of positive integers \\
$\{0,1\}^{n}$ & The set of $n$-dimensional binary column vectors \\
$\llbracket a,b\rrbracket$ & The integer interval $\{a,a+1,\ldots,b\}$, where $a,b\in\mathbb{Z}_{\geq0}$ and $a\leq b$ \\
$\mathbf{1}_n$ & The $n$-dimensional all-ones column vector \\
$\mathbf{0}_n$ & The $n$-dimensional all-zeros column vector \\
$I_n$ & The $n\times n$ identity matrix \\
$\mathbf{e}_{n,j},\bar{\mathbf{e}}_{n,j}$ &
The $n$-dimensional binary column vector $\mathbf{e}_{n,j}$ whose
$j$-th component is $1$ and whose other components are $0$, and its
complement $\bar{\mathbf{e}}_{n,j}:=\mathbf{1}_n-\mathbf{e}_{n,j}$ \\
$\gcd(a,b)$ & The greatest common divisor of integers $a$ and $b$ \\
$[\mathbf{x}]_k$ & The $k$-th component of the vector $\mathbf{x}$ \\
$[z]_+$ & The positive part of $z$, defined by $[z]_{+}=\max(z,0)$, where $z\in\mathbb{R}$ \\
\hline
\end{tabularx}
\end{table}

\section{Problem Formulation}\label{sec:problem formulation}
Consider the following synchronous RNN:
\begin{eqnarray*}
x_i(t+1) & = & f_i(x_1(t),\ldots,x_n(t)),\quad\forall i \in \llbracket 1,n\rrbracket,\\
y_j(t) & = & x_{\phi(j)}(t),\quad\forall j \in \llbracket 1,m\rrbracket,
\end{eqnarray*}
where $x_{i}$ and $y_{j}$ denote the state and output, respectively, and $\phi(j)$ is a function from $\llbracket 1,m\rrbracket$ to $\llbracket 1,n\rrbracket$. Throughout this paper, all vectors are understood to be column vectors unless otherwise specified. For notational simplicity, we write $\xvec=(x_1,\ldots,x_n)$ and $\yvec=(y_1,\ldots,y_m)$ for the column vectors $(x_1,\ldots,x_n)^{\mathsf T}$ and $(y_1,\ldots,y_m)^{\mathsf T}$, respectively. Here, $m$ is the number of observation nodes and $\phi(j)$ denotes the index of the state node selected as the $j$-th observation node.

In this paper, we focus on the following two types of activation functions $f_i(\cdot)$:
\begin{eqnarray*}
&&{\rm \textbf{Linear threshold function:}}\quad f_i(\xvec) = [\avec_{i} \cdot \xvec  \geq \theta_{i}],\\
&&{\rm \textbf{ReLU function:}}\quad f_{i}(\xvec) = \max(\avec_{i} \cdot \xvec + b_{i},0),
\end{eqnarray*}
where $\theta_{i},b_{i}\in\mathbb{R}$, $\avec_{i}\in\mathbb{R}^{n}$, $\avec_{i} \cdot \xvec$ denotes the inner product between two vectors $\avec_{i}$ and $\xvec$, and $[\avec_{i} \cdot \xvec  \geq \theta_{i}]$ is 1 if $\avec_{i} \cdot \xvec  \geq \theta_{i}$ and 0 otherwise. The former case corresponds to a BN with linear threshold
functions, where each node state $x_{i}$ is binary, whereas in the latter case,
each node state $x_{i}$ can take any real value.

We restrict our attention to RNNs in which the update function of each node depends on exactly $K$ relevant state variables. RNNs with linear-threshold and ReLU activation functions are referred to as $K$-LT networks and $K$-ReLU networks, respectively.

For a network $\Sigma$, let $m^{\star}(\Sigma)$ denote the minimum number of observation nodes required for observability. The objective of this paper is to derive upper and lower bounds on $m^{\star}(\Sigma)$ for the classes of $K$-LT and $K$-ReLU networks. More precisely, a set of observation nodes makes the network observable if there exists a finite integer $N\geq 0$ such that any initial state $\xvec(0)$ can be uniquely determined from the output sequence $\yvec(0),\ldots,\yvec(N)$,
where $\yvec(t)=(y_1(t),\ldots,y_m(t))$.

\begin{remark}
The observability considered in this paper is global over the prescribed state space. This is because our aim is to determine bounds on the number of observation nodes that apply to all initial states, rather than only to states near a given state.
\end{remark}

For a prescribed class $\mathcal{C}$ of networks, we derive the following
four types of bounds on $m^{\star}(\Sigma)$:
\begin{description}
    \item[\textbf{General lower bound.}]
    A bound $b_{\mathrm{GL}}$ such that $m^{\star}(\Sigma)\geq b_{\mathrm{GL}}$
    for every $\Sigma\in\mathcal{C}$; equivalently, $b_{\mathrm{GL}}\leq\min_{\Sigma\in\mathcal C}m^{\star}(\Sigma)$.

    \item[\textbf{Best-case upper bound.}]
    A bound $b_{\mathrm{BU}}$ such that there exists a network
    $\Sigma\in\mathcal{C}$ satisfying $m^{\star}(\Sigma)\leq b_{\mathrm{BU}}$; equivalently, $\min_{\Sigma\in\mathcal C}m^{\star}(\Sigma)\leq b_{\mathrm{BU}}$.

    \item[\textbf{Worst-case lower bound.}]
    A bound $b_{\mathrm{WL}}$ such that there exists a network
    $\Sigma\in\mathcal{C}$ satisfying $m^{\star}(\Sigma)\geq b_{\mathrm{WL}}$; equivalently, $b_{\mathrm{WL}}\leq\max_{\Sigma\in\mathcal C}m^{\star}(\Sigma)$.

    \item[\textbf{General upper bound.}]
    A bound $b_{\mathrm{GU}}$ such that $m^{\star}(\Sigma)\leq b_{\mathrm{GU}}$
    for every $\Sigma\in\mathcal{C}$; equivalently, $\max_{\Sigma\in\mathcal C}m^{\star}(\Sigma)\leq b_{\mathrm{GU}}$.
\end{description}
Thus, we have
$b_{\mathrm{GL}}\leq\min_{\Sigma\in\mathcal{C}}m^{\star}(\Sigma)\leq b_{\mathrm{BU}}$
and
$b_{\mathrm{WL}}\leq\max_{\Sigma\in\mathcal{C}}m^{\star}(\Sigma)\leq b_{\mathrm{GU}}$.
The first pair gives the best-case bounds, while the second pair gives the
worst-case bounds. If $b_{\mathrm{GL}}=b_{\mathrm{BU}}$ or
$b_{\mathrm{WL}}=b_{\mathrm{GU}}$, then the corresponding pair of bounds
is said to be tight. These results also show
how the indegree $K$ and the activation function affect the number of
observation nodes needed for observability.

\section{Recurrent Neural Networks with Linear Threshold Functions}\label{sec:linear threshold functions}
In this section, we derive upper and lower bounds on the minimum number of observation nodes required for $K$-LT network observability. We begin with the following two propositions.
\begin{proposition}\label{prop:lt-worst-lb}
There exist a $K$-LT network for which all $n$ state nodes must be observed to ensure observability. Consequently, the worst-case minimum number of observation nodes is $n$.
\end{proposition}
\begin{proposition}\label{prop:lt-best-ub}
For any $n\in\mathbb{Z}_{>0}$ and $K\in\llbracket 2,n\rrbracket$, there exists a $K$-LT network whose minimum number of observation nodes is $\lceil \frac{n}{K} \rceil$.
\end{proposition}
It is known that nested canalyzing functions, including AND and OR functions, can be represented as linear threshold functions \cite{jarrah2007}. Therefore, Propositions \ref{prop:lt-worst-lb} and \ref{prop:lt-best-ub} follow directly from Proposition 5 and Theorem 4 in \cite{sun24}, respectively.

Next, we investigate the best-case minimum number of observation nodes and show that $m=1$ for a broad class of networks, which is not mentioned in \cite{sun24}. First, for the case of $n=K=4$, following the ideas in Example 11 and Theorem 26 of \cite{guo22} but introducing some additional extensions, we consider the following 4-LT network:
\begin{eqnarray*}
x_1(t+1) & = & [x_1(t) - x_2(t) - x_3(t) + 2x_4(t) \geq 1],\\
x_2(t+1) & = & [2.5 x_1(t) + x_2(t) + x_3(t) + x_4(t) \geq 3],\\
x_3(t+1) & = & [x_1(t) + 2x_2(t) - x_3(t) - x_4(t) \geq 1],\\
x_4(t+1) & = & [x_1(t) - x_2(t) + 2x_3(t) - x_4(t) \geq 1].
\end{eqnarray*}

Note that $\xvec(t+1)$ is obtained by applying the right cyclic shift to $\xvec(t)$ except
for the cases of $\xvec(t)=(0,1,1,1)$ and $\xvec(t)=(1,0,0,0)$, and we have the following periodic attractors:
\begin{eqnarray*}
&&(0,0,0,0)\rightarrow(0,0,0,0),\\
&&(1,1,1,1)\rightarrow(1,1,1,1),\\
&&(0,0,1,1)\rightarrow(1,0,0,1)\rightarrow(1,1,0,0)\rightarrow(0,1,1,0)\rightarrow(0,0,1,1),\\
&&(0,1,0,1)\rightarrow(1,0,1,0)\rightarrow(0,1,0,1),\\
&&(0,1,0,0)\rightarrow(0,0,1,0)\rightarrow(0,0,0,1)\rightarrow(1,0,0,0)\rightarrow(1,0,1,1)\\
&&\rightarrow(1,1,0,1)\rightarrow(1,1,1,0)\rightarrow(0,1,1,1)\rightarrow(0,1,0,0).
\end{eqnarray*}
Let $x_4$ be selected as the observation node, so that
$y(t)=x_4(t)$. For any initial state, the corresponding output
sequence is given in Table \ref{tab2}. It follows that $\xvec(0)$ can be
uniquely determined from $y(0),y(1),\ldots,y(4)$.

\begin{table}[htbp]
\centering
\caption{Output sequence for an arbitrary initial state.}\label{tab2}
\begin{tabular}{cccc|ccccc}
\hline
$x_1(0)$ & $x_2(0)$ & $x_3(0)$ & $x_4(0)$ & $y(0)$ & $y(1)$ & $y(2)$ & $y(3)$ & $y(4)$ \\
\hline
0 & 0 & 0 & 0 & 0 & 0 & 0 & 0 & 0 \\
0 & 0 & 0 & 1 & 1 & 0 & 1 & 1 & 0 \\
0 & 0 & 1 & 0 & 0 & 1 & 0 & 1 & 1 \\
0 & 0 & 1 & 1 & 1 & 1 & 0 & 0 & 1 \\
0 & 1 & 0 & 0 & 0 & 0 & 1 & 0 & 1 \\
0 & 1 & 0 & 1 & 1 & 0 & 1 & 0 & 1 \\
0 & 1 & 1 & 0 & 0 & 1 & 1 & 0 & 0 \\
0 & 1 & 1 & 1 & 1 & 0 & 0 & 1 & 0 \\
1 & 0 & 0 & 0 & 0 & 1 & 1 & 0 & 1 \\
1 & 0 & 0 & 1 & 1 & 0 & 0 & 1 & 1 \\
1 & 0 & 1 & 0 & 0 & 1 & 0 & 1 & 0 \\
1 & 0 & 1 & 1 & 1 & 1 & 0 & 1 & 0 \\
1 & 1 & 0 & 0 & 0 & 0 & 1 & 1 & 0 \\
1 & 1 & 0 & 1 & 1 & 0 & 1 & 0 & 0 \\
1 & 1 & 1 & 0 & 0 & 1 & 0 & 0 & 1 \\
1 & 1 & 1 & 1 & 1 & 1 & 1 & 1 & 1 \\\hline
\end{tabular}
\end{table}

The above construction can be generalized for any $K=n>3$, and we have the following proposition.
\begin{proposition}\label{proposition3}
For any $K=n>3$, there exists a $K$-LT network whose minimum number of observation nodes required for observability is one.
\end{proposition}
\begin{proof}
For any $K=n>3$, consider the following $K$-LT network:
\begin{eqnarray*}
x_1(t+1) & = & [x_1(t) - x_2(t) - x_3(t) - \cdots - x_{n-1}(t) + (n-2)x_{n}(t) \geq 1],\\
x_2(t+1) & = & [(n-1.5) x_1(t) + x_2(t) + x_3(t) + \cdots +  x_{n}(t) \geq n-1],\\
x_3(t+1) & = & [x_1(t) + (n-2)x_2(t) - x_3(t) - x_4(t) - \cdots - x_{n-1}(t) - x_{n}(t) \geq 1],\\
x_4(t+1) & = & [x_1(t) - x_2(t) + (n-2) x_3(t) - x_4(t) - \cdots - x_{n-1}(t) - x_{n}(t) \geq 1],\\
&\vdots&\\
x_{n}(t+1) & = & [x_1(t) - x_2(t) - x_3(t) - x_4(t) - \cdots +(n-2) x_{n-1}(t) - x_{n}(t) \geq 1].
\end{eqnarray*}
Let $T_{n}$ denote the state transition map of this network. For any state $\xvec(t)=(x_{1}(t),\ldots,x_{n}(t))\in\{0,1\}^{n}$, one has
\begin{eqnarray*}
\xvec(t+1)=T_{n}(\xvec(t))=\left\{
\begin{array}{ll}
\mathbf{1}_{n}-P_{n}\xvec(t), & {\rm if}~\xvec(t)\in\{\mathbf{e}_{n,1},\bar{\mathbf{e}}_{n,1}\},\\
P_{n}\xvec(t), & {\rm if}~\xvec(t)\notin\{\mathbf{e}_{n,1},\bar{\mathbf{e}}_{n,1}\},
\end{array}
\right.
\end{eqnarray*}
where
$\mathbf{1}_n:=(1,1,\ldots,1),~\mathbf{e}_{n,1}:=(1,0,\ldots,0),~\bar{\mathbf{e}}_{n,1}:=\mathbf{1}_n-\mathbf{e}_{n,1}=(0,1,\ldots,1)$,
and $P_{n}$ is the permutation matrix associated with the right cyclic shift, namely,
\begin{eqnarray*}
P_{n}=
\begin{bmatrix}
0 & 0 &  \cdots & 0 & 1\\
1 & 0 &  \cdots & 0 & 0\\
0 & 1 &  \cdots & 0 & 0\\
\vdots & \vdots & \vdots & \ddots & \vdots\\
0 & 0 &  \cdots & 1 & 0
\end{bmatrix},
\end{eqnarray*}
so that $P_{n}\xvec=(x_{n},x_{1},\ldots,x_{n-1})$ and $(P_{n})^{n}=I_{n}$.

Then, the above transition map gives the following two kinds of periodic attractors.

\textbf{Case 1}: The periodic attractor containing both $\mathbf{e}_{n,1}$ and $\bar{\mathbf{e}}_{n,1}$ is
\begin{eqnarray*}
\mathbf{e}_{n,2}\rightarrow\mathbf{e}_{n,3}\rightarrow\cdots\rightarrow\mathbf{e}_{n,n}\rightarrow\mathbf{e}_{n,1}
\rightarrow\bar{\mathbf{e}}_{n,2}\rightarrow\cdots\rightarrow\bar{\mathbf{e}}_{n,n}\rightarrow\bar{\mathbf{e}}_{n,1}\rightarrow\mathbf{e}_{n,2},
\end{eqnarray*}
where $\mathbf{e}_{n,j}$ is the $n$-dimensional binary column vector whose
$j$-th component is $1$ and whose other components are $0$, and
$\bar{\mathbf{e}}_{n,j}$ denotes its complement, defined by
$\bar{\mathbf{e}}_{n,j}:=\mathbf{1}_n-\mathbf{e}_{n,j}$.

\textbf{Case 2}: For any state $\xvec=(x_{1},\ldots,x_{n})$ not included in the periodic attractor in Case 1, the trajectory
$\xvec, P_{n}\xvec, \ldots, (P_{n})^{n-1}\xvec$ forms a periodic attractor.

Assume that $x_n$ is selected as the observation node, so that $y(t)=x_n(t)$.
For any initial state, the corresponding output
sequence is given in Table \ref{tab3}.
\begin{table}[htbp]
\centering
\caption{Output sequence for an arbitrary initial state.}\label{tab3}
\begin{tabular}{c|c|c}
\hline
{\rm Initial state} $\xvec(0)$ & $y(0),y(1),\ldots,y(n-1)$ & $y(n)$ \\
\hline
$\mathbf{e}_{n,j},~2\le j\le n$ & $\underbrace{0,\ldots,0}_{n-j},1,0,\underbrace{1,\ldots,1}_{j-2}$ & $1-x_{n}(0)$\\
$\mathbf{e}_{n,1}$ & $0,\underbrace{1,\ldots,1}_{n-2},0$ & $1-x_{n}(0)$\\
$\bar{\mathbf{e}}_{n,j},~2\le j\le n$ & $\underbrace{1,\ldots,1}_{n-j},0,1,\underbrace{0,\ldots,0}_{j-2}$ & $1-x_{n}(0)$\\
$\bar{\mathbf{e}}_{n,1}$ & $1,\underbrace{0,\ldots,0}_{n-2},1$ & $1-x_{n}(0)$\\
$\xvec(0)$ {\rm in~Case~2} & $x_{n}(0),x_{n-1}(0),\ldots,x_{1}(0)$ & $x_{n}(0)$\\
\hline
\end{tabular}
\end{table}
Clearly, an initial state in Case 1 can be distinguished from one in
Case 2 by checking whether $y(0)=y(n)$. Within Case 1, only the following
two pairs require further consideration. For
$\mathbf{e}_{n,n}$ and $\bar{\mathbf{e}}_{n,n-1}$, the output sequences
over $t=0,\ldots,n-1$ are
\begin{eqnarray*}
1,0,\underbrace{1,\ldots,1}_{n-2}
\quad\text{and}\quad
1,0,1,\underbrace{0,\ldots,0}_{n-3},
\end{eqnarray*}
respectively. For
$\mathbf{e}_{n,n-1}$ and $\bar{\mathbf{e}}_{n,n}$, the corresponding
output sequences are
\begin{eqnarray*}
0,1,0,\underbrace{1,\ldots,1}_{n-3}
\quad\text{and}\quad
0,1,\underbrace{0,\ldots,0}_{n-2},
\end{eqnarray*}
respectively. Since $n>3$, the two output sequences in each pair are
distinct. Hence, all initial states in Case 1 are distinguishable.

Thus, $\xvec(0)$ can be uniquely determined from the output sequence $y(0),\ldots,y(n)$.
\end{proof}

Furthermore, it can be generalized to the case of $K=4$ and $n=3 \times K = 12$ as follows:
\begin{eqnarray*}
x_{1+4(j+1)}(t+1) & = & [x_{1+4j}(t) - x_{2+4j}(t) - x_{3+4j}(t) + 2x_{4+4j}(t) \geq 1],\\
x_{2+4(j+1)}(t+1) & = & [2.5 x_{1+4j}(t) + x_{2+4j}(t) + x_{3+4j}(t) + x_{4+4j}(t) \geq 3],\\
x_{3+4(j+1)}(t+1) & = & [x_{1+4j}(t) + 2x_{2+4j}(t) - x_{3+4j}(t) - x_{4+4j}(t) \geq 1],\\
x_{4+4(j+1)}(t+1) & = & [x_{1+4j}(t) - x_{2+4j}(t) + 2x_{3+4j}(t) - x_{4+4j}(t) \geq 1]
\end{eqnarray*}
for $j=0,1,2$, where all subscripts are taken modulo 12, with values in
$\llbracket 1,12\rrbracket$.
Let $x_{12}$ be selected as the observation node, so that
$y(t)=x_{12}(t)$. From the output sequence
$y(2),y(5),y(8),y(11),y(14)$, we can uniquely determine the first four components of
$\xvec(0)$.

For example, if the output sequence is
$y(2)=0,y(5)=1,y(8)=1,y(11)=0,y(14)=0$, then the first four components of the initial state are determined as
$\xvec(0)=(0,0,1,1,*,*,*,*,*,*,*,*)$,
where ``$*$'' denotes an unspecified Boolean component, which can be either 0 or 1 and is not used in this partial determination.

\begin{center}
\begin{tabular}{c|cccc|cccc|cccc}
\hline
$t$ & $x_{1}$ & $x_{2}$ & $x_{3}$ & $x_{4}$ & $x_{5}$ & $x_{6}$ & $x_{7}$ & $x_{8}$ & $x_{9}$ & $x_{10}$ & $x_{11}$ & $x_{12}$ \\
\hline
$0$ & 0 & 0 & 1 & 1 & $*$ & $*$ & $*$ & $*$ & $*$ & $*$ & $*$ & $*$ \\
$1$ & $*$ & $*$ & $*$ & $*$ & 1 & 0 & 0 & 1 & $*$ & $*$ & $*$ & $*$ \\
$2$ & $*$ & $*$ & $*$ & $*$ & $*$ & $*$ & $*$ & $*$ & 1 & 1 & 0 & 0 \\
$3$ & 0 & 1 & 1 & 0 & $*$ & $*$ & $*$ & $*$ & $*$ & $*$ & $*$ & $*$ \\
$4$ & $*$ & $*$ & $*$ & $*$ & 0 & 0 & 1 & 1 & $*$ & $*$ & $*$ & $*$ \\
$5$ & $*$ & $*$ & $*$ & $*$ & $*$ & $*$ & $*$ & $*$ & 1 & 0 & 0 & 1 \\
$\vdots$ & $\vdots$ & $\vdots$  & $\vdots$  & $\vdots$  &  $\vdots$ & $\vdots$  & $\vdots$  & $\vdots$  &  $\vdots$ & $\vdots$  & $\vdots$  &  $\vdots$ \\
\hline
\end{tabular}
\end{center}

The above construction can be generalized as follows.
\begin{theorem}\label{prop}
For any $K\geq4$ and $n=h \times K$, where $\gcd(h,K)=1$ and $h$ is odd, there exists a $K$-LT network whose minimum number of observation nodes required for observability is one.
\end{theorem}
\begin{proof}
Consider the following $K$-LT network:
\begin{eqnarray*}
x_{1+K(j+1)}(t+1) & = & [x_{1+Kj}(t) - x_{2+Kj}(t) - \cdots - x_{(K-1)+Kj}(t) + (K-2)x_{K+Kj}(t) \geq 1],\\
x_{2+K(j+1)}(t+1) & = & [(K-1.5) x_{1+Kj}(t) + x_{2+Kj}(t) + \cdots  + x_{K+Kj}(t) \geq K-1],\\
x_{3+K(j+1)}(t+1) & = & [x_{1+Kj}(t) + (K-2)x_{2+Kj}(t)  - \cdots - x_{(K-1)+Kj}(t) - x_{K+Kj}(t) \geq 1],\\
&\vdots&\\
x_{K+K(j+1)}(t+1) & = & [x_{1+Kj}(t) - x_{2+Kj}(t) - \cdots +(K-2) x_{(K-1)+Kj}(t) - x_{K+Kj}(t) \geq 1]
\end{eqnarray*}
for $j\in\llbracket 0,h-1\rrbracket$, where all subscripts are taken modulo $n$, with values in
$\llbracket 1,n\rrbracket$. Here, $\gcd(h,K)=1$ and $h$ is odd.

Partition the state vector into $h$ blocks of length $K$:
\begin{eqnarray*}
\xvec(t)=(\xvec^{(1)}(t),\xvec^{(2)}(t),\ldots,\xvec^{(h)}(t)),
\end{eqnarray*}
where $\xvec^{(j)}(t):=(x_{(j-1)K+1}(t),\ldots,x_{jK}(t))\in\{0,1\}^K,~j\in\llbracket 1,h\rrbracket$.
By the construction of the network, for $j\in\llbracket 1,h\rrbracket$,
\begin{eqnarray*}
\xvec^{(j)}(t+1)=T_{K}(\xvec^{(j-1)}(t))=\left\{
\begin{array}{ll}
\mathbf{1}_{K}-P_{K}\xvec^{(j-1)}(t), & {\rm if}~\xvec^{(j-1)}(t))\in\{\mathbf{e}_{K,1},\bar{\mathbf{e}}_{K,1}\},\\
P_{K}\xvec^{(j-1)}(t), & {\rm if}~\xvec^{(j-1)}(t)\notin\{\mathbf{e}_{K,1},\bar{\mathbf{e}}_{K,1}\},
\end{array}
\right.
\end{eqnarray*}
where $\xvec^{(0)}(t):=\xvec^{(h)}(t)$.

Assume that $x_{hK}$ is selected as the observation node. The state information of the first block, $\xvec^{(1)}(0)=(x_1(0),x_2(0),\ldots,x_K(0))$,
is propagated to the $h$-th block at the time instants
$t=lh-1,~l\in\mathbb Z_{>0}$.
More precisely, $\xvec^{(h)}(lh-1)=T_K^{lh-1}(\xvec^{(1)}(0))$.
Let $z_l:=x_{hK}(lh-1),~l\in\llbracket1,K+1\rrbracket$.
We show that the initial state $\xvec^{(1)}(0)$ can be uniquely determined from the output sequence
$z_1,z_2,\ldots,z_{K+1}$.

Hereafter, $y(0),y(1),\ldots,y(K)$ refer to the output sequence of the $K$-LT network considered in Proposition~\ref{proposition3}.

We first consider the special case $h=1$. In this case, $z_l=x_K(l-1),~l\in\llbracket1,K+1\rrbracket$.
It follows from Table~\ref{tab3} that the output sequence
$y(0),y(1),\ldots,y(K-1)$
uniquely determines the initial state within each of the two cases in Proposition~\ref{proposition3}. Moreover, the additional output $y(K)$ distinguishes the two cases. Therefore, $\xvec^{(1)}(0)$ is uniquely determined by
$z_1,z_2,\ldots,z_{K+1}$ when $h=1$.

We next consider the case $h>1$. Since $\gcd(h,K)=1$, the residues
$r_l:=(lh-1)\bmod K,~l\in\llbracket1,K\rrbracket$,
run through all the elements of $\llbracket0,K-1\rrbracket$. For each
$l\in\llbracket1,K\rrbracket$, we have
$(lh-1)\bmod 2K=r_l+\varepsilon_lK,~\varepsilon_l\in\{0,1\}$,

Suppose that $\xvec^{(1)}(0)$ belongs to the periodic attractor described in Case~1 of Proposition~\ref{proposition3}. Note that in this case the output sequence $y(t)$ in Proposition~\ref{proposition3} satisfies
$y(t+K)=1-y(t)$.
Hence,
\begin{eqnarray*}
z_l=\left\{
\begin{array}{ll}
y(r_l), & \varepsilon_l=0,\\
1-y(r_l), & \varepsilon_l=1.
\end{array}
\right.
\end{eqnarray*}
Since $r_1,\ldots,r_K$ form a permutation of
$\llbracket0,K-1\rrbracket$, the output sequence
$z_1,\ldots,z_K$ uniquely determine the output values
$y(0),y(1),\ldots,y(K-1)$
listed in Table~\ref{tab3}. According to Table~\ref{tab3}, these $K$ output values uniquely determine the initial state among all states in Case~1 of Proposition~\ref{proposition3}. Therefore, $z_1,\ldots,z_K$ uniquely determine $\xvec^{(1)}(0)$ within Case~1 of Proposition~\ref{proposition3}.

Suppose that $\xvec^{(1)}(0)$ belongs to the periodic attractor described in Case~2 of Proposition~\ref{proposition3}.
Then
\begin{eqnarray*}
\left[T_K^s\bigl(\xvec^{(1)}(0)\bigr)\right]_K=x_{K-(s\bmod K)}(0),\quad s\in\mathbb Z_{\geq0},
\end{eqnarray*}
where $[\cdot]_{K}$ denotes the $K$-th component of a $K$-dimensional vector.
Consequently,
\begin{eqnarray*}
z_l=x_{K-((lh-1)\bmod K)}(0),\quad l\in\mathbb Z_{>0}.
\end{eqnarray*}
Here, $K-r_l,~l\in\llbracket1,K\rrbracket$,
run through all the elements of $\llbracket1,K\rrbracket$.
Since $z_l=x_{K-r_l}(0)$, the output sequence $z_1,\ldots,z_K$ contain every component of $\xvec^{(1)}(0)=(x_1(0),x_2(0),\ldots,x_K(0))$
exactly once. Moreover, the order of these components is determined by $h$. Therefore, $z_1,\ldots,z_K$ uniquely determine $\xvec^{(1)}(0)$ within Case~2 of Proposition~\ref{proposition3}.

In summary, within each case, the output sequence $z_1,\ldots,z_K$
is related to the corresponding output sequence in Table~\ref{tab3} of Proposition~\ref{proposition3} by a bijective transformation determined by $h$. Thus, $z_1,\ldots,z_K$ uniquely determine the initial state $\xvec^{(1)}(0)$ within each case.

It remains to distinguish the two cases. Since $h$ is odd, then $Kh\equiv K\pmod{2K}$, and hence $(K+1)h-1\equiv h-1+K \pmod{2K}$.
Therefore,
\begin{eqnarray*}
z_{K+1}=\left\{
\begin{array}{ll}
1-z_{1}, & {\rm if}~\xvec^{(1)}(0) {\rm~belongs~to~Case~1},\\
z_{1}, & {\rm if}~\xvec^{(1)}(0) {\rm~belongs~to~Case~2}.
\end{array}
\right.
\end{eqnarray*}
Thus, the two cases can be uniquely distinguished by comparing $z_{K+1}$ with $z_1$.
Hence, the output sequence $z_1,\ldots,z_{K+1}$ uniquely determines the initial state $\xvec^{(1)}(0)$.

The same argument applies to the other blocks. More precisely, for $j\in\llbracket 1,h\rrbracket$, the block $\xvec^{(j)}(0)$ reaches the observation node at the time instants $t=lh-j,~l\in\mathbb Z_{>0}$.
The shift by $j$ does not affect the above residue arguments. Hence, each
initial block $\xvec^{(j)}(0)$ can be uniquely determined from the corresponding
subsequence of the output. Therefore, the entire initial state $\xvec(0)$ can
be uniquely determined from $x_{Kh}(t),~t\in\llbracket 0,(K+1)h-1\rrbracket$.
\end{proof}

\begin{remark}
For the above Proposition \ref{prop}, if $h$ is even, the initial states in the two cases cannot always be distinguished. For example, let $K=7$ and $h=2$. The two initial blocks $\xvec^{(1)}(0)=(0,1,0,0,0,0,0)$ and $\widetilde{\xvec}^{(1)}(0)=(1,1,1,0,1,0,1)$,
which belong to Case~1 and Case~2, respectively, generate identical output sequences $0,0,1,1,1,1,1,0,0,1,1,1,1,1,\ldots$.
\end{remark}
\section{Recurrent Neural Networks with ReLU Functions}\label{sec:relu functions}
In this section, we consider three subtypes of ReLU networks.

\textbf{Binary ReLU network}: The state of each node is restricted to binary values, i.e., $\xvec(t)\in\{0,1\}^n$.

\textbf{Positive ReLU network}: The state remains in the nonnegative orthant, i.e., $\xvec(t)\in\mathbb{R}_{\geq0}^n$.

\textbf{General ReLU network}: The state can take any value in $\mathbb{R}^n$.
\subsection{Binary ReLU Networks}
Consider an AND function $f(x_1,x_2,\ldots,x_K)=l_1 \land l_2 \land \cdots \land l_K$,
where each $l_i$ is a literal, i.e., either $x_i$ or $\lnon{x_i}$. Define
\begin{eqnarray*}
lit(l_i) & = & \left\{
\begin{array}{ll}
x_i & \mbox{if $l_i=x_i$},\\
1-x_i & \mbox{if $l_i=\lnon{x_i}$}.
\end{array}
\right.
\end{eqnarray*}
Then
\begin{eqnarray*}
f(x_1,x_2,\ldots,x_K) & = & \max\left(\sum_{i=1}^{K} lit(l_i) - K + 1, 0\right).
\end{eqnarray*}
Therefore, any $K$-AND function consisting of $K$ literals can be represented as a binary $K$-ReLU function.
Furthermore, we can prove that the converse of this property also holds.
\begin{proposition}
Every $K$-AND BN admits an equivalent representation as a binary $K$-ReLU network, and conversely.
\end{proposition}
\begin{proof}
The forward implication has already been established. It remains to prove the converse.

Suppose that exactly $H$ of the $K$ coefficients in $f(x_1,\ldots,x_K)$ are positive, where $H\in\llbracket 0,K\rrbracket$. Without loss of generality, assume that the coefficients of $x_1,\ldots,x_H$ are positive and those of
$x_{H+1},\ldots,x_K$ are negative, and let
$f(x_1,\ldots,x_K)=\max(a_{1}x_{1}+\cdots+a_{H}x_{H}-a_{H+1}x_{H+1}-\cdots-a_{K}x_{K}+b,0)$,
where $a_{1},\ldots,a_{K}>0$. We note that $a_{1}x_{1}+\cdots+a_{H}x_{H}-a_{H+1}x_{H+1}-\cdots-a_{K}x_{K}+b$ has a unique maximum value $a_{1}+\cdots+a_{H}+b$ when $x_1=\cdots=x_H=1$ and $x_{H+1}=\cdots=x_K=0$. Since $f(x_1,\ldots,x_K)$ is a binary $K$-ReLU function, $a_{1}+\cdots+a_{H}+b$ can take the value 1 or 0.

If $a_{1}+\cdots+a_{H}+b=1$, which means that
\begin{eqnarray*}
&&f(x_1,\ldots,x_K)= \left\{
\begin{array}{ll}
1 & \mbox{if $x_1=\cdots=x_H=1,x_{H+1}=\cdots=x_K=0$},\\
0 & \mbox{otherwise},
\end{array}
\right.
\end{eqnarray*}
then this binary $K$-ReLU function $f(x_1,\ldots,x_K)$ is equivalent to the following $K$-AND function
\begin{eqnarray*}
x_1 \land \cdots \land x_H \land \lnon{x_{H+1}} \land \cdots \land \lnon{x_{K}}.
\end{eqnarray*}

If $a_{1}+\cdots+a_{H}+b=0$, which means that $f(x_1,\ldots,x_K)$ is identically zero,
then this binary $K$-ReLU function $f(x_1,\ldots,x_K)$ becomes meaningless in this case.
\end{proof}

Therefore, by the equivalence between $K$-AND BNs whose update functions are conjunctions of $K$ literals and the corresponding binary $K$-ReLU networks, the results established in Proposition 5 and Theorems 2,3 of \cite{sun24} for $K$-AND BNs also hold for the corresponding binary $K$-ReLU networks.

Specifically, for any binary $K$-ReLU network with $n$ nodes, the general lower bound, the best case upper bound, and the worst case lower bound for the minimum number of observation nodes are $\left[(1-K)+\frac{2^{K}-1}{2^{K}}\log_{2}(2^{K}-1)\right]n$, $\left(\frac{2^{K}-K-1}{2^{K}-1}\right)n$, and $n$, respectively.
\subsection{Positive ReLU Networks}
We next consider positive ReLU networks. Under the nonnegativity assumptions introduced below, the input to the ReLU function remains nonnegative. Hence, the positive ReLU network reduces to a positive affine system. Its observability can therefore be studied using standard linear-system theory. This case serves as a reference for the analysis of general ReLU networks.

\begin{proposition}
Consider the positive ReLU network
\begin{eqnarray*}
\xvec(t+1) & = & \max\left(A\xvec(t)+\bvec,\mathbf{0}_{n}\right),\\
\yvec(t) & = & C\xvec(t),
\end{eqnarray*}
where $\xvec(0)\in\mathbb{R}_{\geq0}^n$, $A\in\mathbb{R}_{\geq0}^{n\times n}$, $\bvec\geq\mathbb{R}_{\geq0}^{n}$, $\yvec(t)\in\mathbb{R}^p$, $C\in\mathbb{R}^{p\times n}$, and the maximum is taken componentwise. Then the initial state $\xvec(0)$ can be uniquely determined from the output sequence $\yvec(0),\yvec(1),\ldots,\yvec(n-1)$ if and only if
\begin{eqnarray*}
\rank
\begin{bmatrix}
C \\
CA \\
\vdots\\
CA^{n-1}
\end{bmatrix}
=n.
\end{eqnarray*}
\end{proposition}
\begin{proof}
Since $\xvec(0)\in\mathbb{R}_{\geq0}^n$, $A\geq\mathbf{0}_{n\times n}$, and $\bvec\geq \mathbf{0}_{n}$, it follows
that $A\xvec(t)+\bvec\geq \mathbf{0}_{n}$. Therefore, the positive ReLU network
reduces to the positive linear system $\xvec(t+1)=A\xvec(t)+\bvec$.
It follows that $\xvec(t)=A^{t}\xvec(0)+\sum_{j=0}^{t-1}A^j\bvec$,
and hence $\yvec(t)=CA^{t}\xvec(0)+\sum_{j=0}^{t-1}CA^j\bvec$.
The second term is completely determined by $A$, $\bvec$, and $C$. Therefore, the
initial state can be uniquely determined from the output sequence if and only if
the classical observability matrix of the pair $(C,A)$ has full column rank.
\end{proof}

\begin{remark}
For a positive ReLU network, the initial state $\xvec(0)$ can be uniquely determined by observing a single node $x_i(t),~i\in\llbracket1,n\rrbracket$, if and only if
$\rank\begin{bmatrix}C_{i} \\C_{i}A \\\vdots\\C_{i}A^{n-1}\end{bmatrix}=n$,
where $C_{i}=\mathbf{e}_{n,i}^{\mathsf T}$.
\end{remark}
\subsection{General ReLU Networks}
We first consider the following ReLU network with $K=1$ and $n=2^h-1$, where $h\geq2$:
\begin{eqnarray*}
x_1(t+1) & = & \max(x_2(t),0),\\
x_{2i}(t+1) & = & \max(x_i(t),0),\\
x_{2i+1}(t+1) & = & \max(-x_i(t),0),\quad i\in\llbracket1,2^{h-1}-1\rrbracket.
\end{eqnarray*}
Let $x_{2j}(t)$ and $x_{2j+1}(t)$, where $j\in\llbracket2^{h-2},2^{h-1}-1\rrbracket$, be selected as the observation nodes.

At $t=0$, $x_i(0),~i\in\llbracket2^{h-1},2^{h}-1\rrbracket$ are obtained directly from the observation nodes.

At $t=1$, $x_i(0),~i\in\llbracket2^{h-2},2^{h-1}-1\rrbracket$ are recovered as $x_i(0)=x_{2i}(1)-x_{2i+1}(1)$.

At $t=2$, $x_i(0),~i\in\llbracket2^{h-3},2^{h-2}-1\rrbracket$ are recovered as
$x_i(0)=x_{2i}(1)-x_{2i+1}(1)=(x_{4i}(2)-x_{4i+1}(2))-(x_{4i+2}(2)-x_{4i+3}(2))=x_{4i}(2)-x_{4i+2}(2)$, where $x_{4i+1}(2)=x_{4i+3}(2)=0$.

More generally, for each $t\in\llbracket1,h-1\rrbracket$, $x_i(0),~i\in\llbracket2^{h-1-t},2^{h-t}-1\rrbracket$ are recovered according to
\begin{eqnarray*}
x_i(0)=x_{2^t i}(t)-x_{2^t i+2^{t-1}}(t).
\end{eqnarray*}
Therefore, initial state $\xvec(0)$ can be uniquely recovered from the output sequence
$x_{2j}(t),x_{2j+1}(t),~j\in\llbracket2^{h-2},2^{h-1}-1\rrbracket$ over
$t\in\llbracket0,h-1\rrbracket$. In this case, $\frac{n+1}{2}$ observation nodes are sufficient to recover the entire initial state.

Next, we consider the following ReLU network with $K=2$ and $n=4h$, where $h\geq1$:
\begin{eqnarray*}
x_{4i+1}(t+1) & = & \max(x_{4i+1}(t)+x_{4i+2}(t),0),\\
x_{4i+2}(t+1) & = & \max(-x_{4i+1}(t)-x_{4i+2}(t),0),\\
x_{4i+3}(t+1) & = & \max(x_{4i+1}(t)-x_{4i+2}(t),0),\\
x_{4i+4}(t+1) & = & \max(-x_{4i+1}(t)+x_{4i+2}(t),0),\quad i\in\llbracket0,h-1\rrbracket.
\end{eqnarray*}
Let $x_{4i+3}(t)$ and $x_{4i+4}(t)$, where $i\in\llbracket0,h-1\rrbracket$, be selected as the observation nodes.

At $t=0$, $x_{4i+3}(0),x_{4i+4}(0),~i\in\llbracket0,h-1\rrbracket$ are obtained directly from the observation nodes.

At $t=1$, one has $x_{4i+3}(1)-x_{4i+4}(1)=x_{4i+1}(0)-x_{4i+2}(0)$.

At $t=2$, one has $x_{4i+3}(2)-x_{4i+4}(2)=x_{4i+1}(1)-x_{4i+2}(1)=x_{4i+1}(0)+x_{4i+2}(0)$.

Thus, initial states $x_{4i+1}(0),x_{4i+2}(0),~i\in\llbracket0,h-1\rrbracket$ are recovered as $x_{4i+1}(0)=\frac{1}{2}(x_{4i+3}(2)-x_{4i+4}(2)+x_{4i+3}(1)-x_{4i+4}(1))$ and $x_{4i+2}(0)=\frac{1}{2}(x_{4i+3}(2)-x_{4i+4}(2)-x_{4i+3}(1)+x_{4i+4}(1))$.

Therefore, initial state $\xvec(0)$ can be uniquely recovered from the output sequence $x_{4i+3}(t),x_{4i+4}(t),~i\in\llbracket0,h-1\rrbracket$ over
$t=0,1,2$. In this case, $\frac{n}{2}$ observation nodes are sufficient to recover the entire initial state.

We next consider the case $K=3$ and establish the following proposition.
\begin{proposition}\label{proposition4}
For $n=6$, there exists a general 3-ReLU network that is observable from $m=3$ observation nodes.
\end{proposition}
\begin{proof}
Consider the following general 3-ReLU network:
\begin{eqnarray*}
x_1(t+1) & = & \max(-x_1(t)+x_2(t)+x_3(t),0),\\
x_2(t+1) & = & \max(x_1(t)-x_2(t)-x_3(t),0),\\
x_3(t+1) & = & \max(x_1(t)-x_2(t)+x_3(t),0),\\
x_4(t+1) & = & \max(-x_1(t)+x_2(t)-x_3(t),0),\\
x_5(t+1) & = & \max(x_1(t)+x_2(t)-x_3(t),0),\\
x_6(t+1) & = & \max(-x_1(t)-x_2(t)+x_3(t),0),\\
y_1(t) & = & x_4(t),\\
y_2(t) & = & x_5(t),\\
y_3(t) & = & x_6(t).
\end{eqnarray*}
For $i\in\llbracket1,3\rrbracket$, let $s_i(t):=\sum_{j=1}^{3}x_j(t)-2x_i(t)$.
Then we have
\begin{eqnarray*}
x_1(0)=\frac{s_2(0)+s_3(0)}{2},\quad
x_2(0)=\frac{s_1(0)+s_3(0)}{2},\quad
x_3(0)=\frac{s_1(0)+s_2(0)}{2}.
\end{eqnarray*}
Since $x_4,~x_5,~x_6$ are selected as observation nodes,
 $x_4(0),x_5(0),x_6(0)$ are directly
available from the output sequence at $t=0$. Therefore, recovering the entire
initial state reduces to determining $s_1(0),s_2(0),s_3(0)$.

To this end, define $[z]_+:=\max(z,0),~z\in\mathbb{R}$. Then we have $x_{2i-1}(t+1)=[s_{i}(t)]_{+}$ and $x_{2i}(t+1)=[-s_{i}(t)]_{+}$, where $i\in\llbracket1,3\rrbracket$.
According to $[z]_{+}-[-z]_{+}=z$ and $[z]_{+}+[-z]_{+}=|z|$,
we have
\begin{eqnarray*}
s_1(t+1)&=&-[s_1(t)]_{+}+[-s_1(t)]_{+} +[s_2(t)]_+=-s_1(t)+[s_2(t)]_+,\\
s_2(t+1)&=&[s_1(t)]_{+}-[-s_1(t)]_{+} +[s_2(t)]_+=s_1(t)+[s_2(t)]_+,\\
s_3(t+1)&=&[s_1(t)]_{+}+[-s_1(t)]_{+} -[s_2(t)]_+=|s_1(t)|-[s_2(t)]_+.
\end{eqnarray*}
First, we obtain
$s_3(0)=[s_3(0)]_{+}-[-s_3(0)]_{+}=x_5(1)-x_6(1)$.
Moreover, we have
\begin{eqnarray}
x_4(2)  =  [-s_2(1)]_+ =  [-s_1(0)-[s_2(0)]_+]_+,
\end{eqnarray}
\begin{eqnarray}
x_5(2)-x_6(2)=s_3(1)=|s_1(0)|-[s_2(0)]_+,
\end{eqnarray}
\begin{eqnarray}
x_5(3)-x_6(3)=s_3(2)=|-s_1(0)+[s_2(0)]_+|-[s_1(0)+[s_2(0)]_+]_+.
\end{eqnarray}

If $x_4(2)>0$, i.e., $-s_{1}(0)>[s_2(0)]_+$, then $x_4(2)=-s_1(0)-[s_2(0)]_+,~x_5(2)-x_6(2)=-s_1(0)-[s_2(0)]_+$, and $x_5(3)-x_6(3)=-s_1(0)+[s_2(0)]_+$. Hence, we have
\begin{eqnarray*}
s_1(0)&=&-\frac{x_{4}(2)+x_{5}(3)-x_{6}(3)}{2},\\
s_{2}(0)&=&\frac{x_5(3)-x_6(3)-x_4(2)}{2}-x_4(1).
\end{eqnarray*}

If $x_4(2)=0$ and $x_{5}(2)-x_{6}(2)>0$, i.e., $-s_{1}(0)\leq[s_2(0)]_+<s_{1}(0)$, then $x_5(2)-x_6(2)=s_1(0)-[s_2(0)]_+$, and $x_5(3)-x_6(3)=-2[s_2(0)]_+$. Hence, we have
\begin{eqnarray*}
s_1(0)&=&(x_{5}(2)-x_{6}(2))-\frac{x_{5}(3)-x_{6}(3)}{2},\\
s_{2}(0)&=&-\frac{x_5(3)-x_6(3)}{2}-x_4(1).
\end{eqnarray*}

If $x_4(2)=0$ and $x_{5}(2)-x_{6}(2)\leq0$, i.e., $-s_{1}(0)\leq[s_2(0)]_+$ and $|s_{1}(0)|=-s_{1}(0)$, then $x_5(2)-x_6(2)=-s_1(0)-[s_2(0)]_+$, and $x_5(3)-x_6(3)=-2s_1(0)$. Hence, we have
\begin{eqnarray*}
s_1(0)&=&-\frac{x_{5}(3)-x_{6}(3)}{2},\\
s_{2}(0)&=&-(x_{5}(2)-x_{6}(2))+\frac{x_5(3)-x_6(3)}{2}-x_4(1).
\end{eqnarray*}

Therefore, when $x_{4}(2)>0$, we have
\begin{eqnarray*}
x_1(0)
&=&\frac{x_5(3)-x_6(3)-x_4(2)}{4}
-\frac{x_4(1)}{2}
+\frac{x_5(1)-x_6(1)}{2},\\
x_2(0)
&=&-\frac{x_4(2)+x_5(3)-x_6(3)}{4}
+\frac{x_5(1)-x_6(1)}{2},\\
x_3(0)
&=&-\frac{x_4(2)+x_4(1)}{2},
\end{eqnarray*}
when $x_{4}(2)=0$, we have
\begin{eqnarray*}
x_1(0)
&=&-\frac{x_{4}(1)}{2}+\frac{x_5(1)-x_6(1)}{2}+\frac{[x_{6}(2)-x_{5}(2)]_{+}}{2}+\frac{|x_5(3)-x_6(3)|}{4},\\
x_2(0)
&=&\frac{x_5(1)-x_6(1)}{2}+\frac{[x_{5}(2)-x_{6}(2)]_{+}}{2}-\frac{x_{5}(3)-x_{6}(3)}{4},\\
x_3(0)
&=&-\frac{x_{4}(1)}{2}+\frac{|x_{5}(2)-x_{6}(2)|}{2}+\frac{[-x_{5}(3)+x_{6}(3)]_{+}}{2},
\end{eqnarray*}
which means that the entire initial state $\xvec(0)$ can be uniquely recovered from
the output sequence $y_1(t),y_2(t),y_3(t)$ over $t\in\llbracket0,3\rrbracket$.
\end{proof}

We next generalize the above result to the case $K\geq 3$.

\begin{theorem}\label{theorem5}
For $n=2K,~K\geq3$, there exists a general $K$-ReLU network that is observable from $m=\frac{n}{2}$ observation nodes.
\end{theorem}
\begin{proof}
Consider the following general $K$-ReLU network:
\begin{eqnarray*}
x_{2i-1}(t+1) & = & \max(s_{i}(t),0),\\
x_{2i}(t+1) & = & \max(-s_{i}(t),0),\\
y_{1}(t) & = & x_{K+1}(t),\\
y_{2}(t) & = & x_{K+2}(t),\\
&\vdots&\\
y_{K}(t) & = & x_{2K}(t),
\end{eqnarray*}
where $i\in\llbracket1,K\rrbracket$ and $s_{i}(t):=\sum_{j=1}^{K}x_j(t)-2x_i(t)$.

First, since $x_{K+1},\ldots,x_{2K}$ are selected as observation nodes,
their initial values $x_{K+1}(0),\ldots,x_{2K}(0)$ are directly
available from the output sequence at $t=0$.

On the other hand, we have $(s_1(0),\ldots,s_K(0))=S_K(x_1(0),\ldots,x_K(0))$,
where
\begin{eqnarray*}
S_K
=
\mathbf{1}_K\mathbf{1}_K^{\mathsf T}-2I_K
=
\begin{bmatrix}
-1 & 1 & \cdots & 1\\
1 & -1 & \cdots & 1\\
\vdots & \vdots & \ddots & \vdots\\
1 & 1 & \cdots & -1
\end{bmatrix}.
\end{eqnarray*}
For $K\geq 3$, $S_K$ is nonsingular, and
$S_K^{-1}=-\frac{1}{2}I_K+\frac{1}{2(K-2)}\mathbf{1}_K\mathbf{1}_K^{\mathsf T}$.
Therefore, we have
\begin{eqnarray*}
(x_1(0),\ldots,x_K(0))=S_K^{-1}(s_1(0),\ldots,s_K(0)).
\end{eqnarray*}
Hence, recovering the entire
initial state reduces to determining $s_1(0),\ldots,s_{K}(0)$.

To this end, define $[z]_+:=\max(z,0),~z\in\mathbb{R}$.
Clearly, we have $[z]_{+}-[-z]_{+}=z$ and $[z]_{+}+[-z]_{+}=|z|$.

Suppose that $K$ is even and let $q:=\frac{K}{2}$, one has
\begin{eqnarray*}
s_{2i-1}(t+1)&=&-s_i(t)+\sum_{\substack{j=1\\j\neq i}}^{q}|s_j(t)|,\\
s_{2i}(t+1)&=&s_i(t)+\sum_{\substack{j=1\\j\neq i}}^{q}|s_j(t)|.
\end{eqnarray*}
For $i\in\llbracket q+1,K\rrbracket$, we have $s_i(t)=x_{2i-1}(t+1)-x_{2i}(t+1)$,
and hence $s_i(t),~i\in\llbracket q+1,K\rrbracket$ is directly determined from the observation nodes.

Moreover, for $i\in\llbracket 1,q\rrbracket$,
\begin{eqnarray}\label{e4}
s_i(t)=\frac{s_{2i}(t+1)-s_{2i-1}(t+1)}{2}.
\end{eqnarray}
Therefore, for each $i\in\llbracket 2,q\rrbracket$, $s_i(0)$ can be recovered recursively using (\ref{e4}).  However, (\ref{e4}) cannot be used directly to recover $s_1(0)$,
since, for $i=1$ and $t=0$, its right-hand side involves the unknown value $s_1(1)$. After $d$ recursive steps, the resulting indices lie in $\llbracket (i-1)2^d+1,i2^d\rrbracket$.
Hence, all the resulting indices belong to $\llbracket q+1,K\rrbracket$ once
$(i-1)2^d\ge q$.
The minimum number of recursive steps required to recover $s_i(0)$ is therefore
$d_i=\left\lceil\log_2\frac{q}{i-1}\right\rceil$.
The corresponding terminal values are directly obtained from the observation nodes at time $d_i+1$. Since $d_i$ is maximized at $i=2$, all the values
$s_2(0),s_3(0),\ldots,s_q(0)$
can be recovered from the observation nodes up to time
$\left\lceil\log_2q\right\rceil+1$.
Finally, since $s_2(1)=s_1(0)+\sum_{j=2}^{q}|s_j(0)|$,
we obtain $s_1(0)=s_2(1)-\sum_{j=2}^{q}|s_j(0)|$.

Therefore, $s_1(0),\ldots,s_q(0)$ can be uniquely recovered
from the output sequence  $y_{1}(t),\ldots,y_{K}(t)$ over $t\in\llbracket 0,\lceil \log_{2}(q) \rceil+1\rrbracket$.

Similarly, suppose that $K$ is odd and let $q:=\frac{K-1}{2}$. Then, for $i\in\llbracket 1,q\rrbracket$, we have
\begin{eqnarray}\label{e5}
s_{2i-1}(t+1) = -s_i(t)+\sum_{\substack{j=1\\j\neq i}}^{q}|s_j(t)|+[s_{q+1}(t)]_{+},
\end{eqnarray}
\begin{eqnarray}\label{e6}
s_{2i}(t+1) = s_i(t)+\sum_{\substack{j=1\\j\neq i}}^{q}|s_j(t)|+[s_{q+1}(t)]_{+},
\end{eqnarray}
and $s_{K}(t+1)=\sum_{j=1}^{q}|s_{j}(t)|-[s_{q+1}(t)]_{+}$. We also have $s_i(t)=\frac{s_{2i}(t+1)-s_{2i-1}(t+1)}{2}$.

Substituting $i=q$ into Eqs. (\ref{e5}) and (\ref{e6}), we obtain
\begin{eqnarray*}
s_{K-2}(1)&=&-s_q(0)+\sum_{j=1}^{q-1}|s_j(0)|+[s_{q+1}(0)]_+,\\
s_{K-1}(1)&=&s_q(0)+\sum_{j=1}^{q-1}|s_j(0)|+[s_{q+1}(0)]_+.
\end{eqnarray*}
Hence, $\max(s_{K-2}(1),s_{K-1}(1))=\sum_{j=1}^{q}|s_j(0)|+[s_{q+1}(0)]_+$.
Together with $s_K(1)=\sum_{j=1}^{q}|s_j(0)|-[s_{q+1}(0)]_+$,
this gives
\begin{eqnarray*}
[s_{q+1}(0)]_+=\frac{\max(s_{K-2}(1),s_{K-1}(1))-s_K(1)}{2}.
\end{eqnarray*}
Since
$x_{K+1}(1)=[-s_{q+1}(0)]_+$,
we obtain
\begin{eqnarray*}
s_{q+1}(0)=
\frac{\max(s_{K-2}(1),s_{K-1}(1))-s_K(1)}{2}
-x_{K+1}(1),
\end{eqnarray*}
where $s_{K-2}(1)=x_{2K-5}(2)-x_{2K-4}(2),~s_{K-1}(1)=x_{2K-3}(2)-x_{2K-2}(2)$, and
$s_K(1)=x_{2K-1}(2)-x_{2K}(2)$.
If $2K-5\geq K+1$, i.e., $K\geq6$, then $s_{q+1}(0)$ can be uniquely recovered from the observation nodes.

Consequently, for $i\in\llbracket q+2,K\rrbracket$, we have $s_{i}(t)=x_{2i-1}(t+1)-x_{2i}(t+1)$,
hence, $s_{q+2}(t),\ldots,s_K(t)$ can be uniquely recovered
from the observation nodes.

For each $i\in\llbracket2,q\rrbracket$, repeated application of (\ref{e4}) shows that
$s_2(0),s_3(0),\ldots,s_q(0)$
can be recovered from the observations up to time
$\left\lceil\log_2 q\right\rceil+1$.
Finally,
$s_1(0)=s_2(1)-\sum_{j=2}^{q}|s_j(0)|-[s_{q+1}(0)]_+$.

Therefore, $s_1(0),\ldots,s_q(0)$ can be uniquely recovered
from the output sequence $y_{1}(t),\ldots,y_{K}(t)$ over $t\in\llbracket 0,\lceil \log_{2}(q) \rceil+1\rrbracket$.

The case $K=5$ is to be established in Proposition \ref{proposition5}.

Combining the above cases, for every $K\geq 3$, the entire initial state
$\xvec(0)$ can be uniquely recovered from the output sequence $y_{1}(t),\ldots,y_{K}(t)$ over $t\in\llbracket 0,\lceil \log_{2}(q) \rceil+1\rrbracket$.
\end{proof}

\begin{proposition}\label{proposition5}
For $n=10$, there exists a general 5-ReLU network that is observable from $m=5$ observation nodes.
\end{proposition}
\begin{proof}
Consider the following general 5-ReLU network:
\begin{eqnarray*}
x_1(t+1) & = & \max(-x_1(t)+x_2(t)+x_3(t)+x_{4}(t)+x_{5}(t),0),\\
x_2(t+1) & = & \max(x_1(t)-x_2(t)-x_3(t)-x_{4}(t)-x_{5}(t),0),\\
x_3(t+1) & = & \max(x_1(t)-x_2(t)+x_3(t)+x_{4}(t)+x_{5}(t),0),\\
x_4(t+1) & = & \max(-x_1(t)+x_2(t)-x_3(t)-x_{4}(t)-x_{5}(t),0),\\
x_5(t+1) & = & \max(x_1(t)+x_2(t)-x_3(t)+x_{4}(t)+x_{5}(t),0),\\
x_6(t+1) & = & \max(-x_1(t)-x_2(t)+x_3(t)-x_{4}(t)-x_{5}(t),0),\\
x_7(t+1) & = & \max(x_1(t)+x_2(t)+x_3(t)-x_{4}(t)+x_{5}(t),0),\\
x_8(t+1) & = & \max(-x_1(t)-x_2(t)-x_3(t)+x_{4}(t)-x_{5}(t),0),\\
x_9(t+1) & = & \max(x_1(t)+x_2(t)+x_3(t)+x_{4}(t)-x_{5}(t),0),\\
x_{10}(t+1) & = & \max(-x_1(t)-x_2(t)-x_3(t)-x_{4}(t)+x_{5}(t),0),\\
y_1(t) & = & x_6(t),\\
y_2(t) & = & x_7(t),\\
y_3(t) & = & x_8(t),\\
y_4(t) & = & x_9(t),\\
y_5(t) & = & x_{10}(t).
\end{eqnarray*}
For $i\in\llbracket1,5\rrbracket$, let $s_i(t):=\sum_{j=1}^{5}x_j(t)-2x_i(t)$.
Then
\begin{eqnarray*}
s_1(t+1) & = & -s_1(t)+|s_2(t)|+[s_3(t)]_{+},\\
s_2(t+1) & = & s_1(t)+|s_2(t)|+[s_3(t)]_{+},\\
s_3(t+1) & = & -s_2(t)+|s_1(t)|+[s_3(t)]_{+},\\
s_4(t+1) & = & s_2(t)+|s_1(t)|+[s_3(t)]_{+},\\
s_5(t+1) & = & |s_1(t)|+|s_2(t)|-[s_3(t)]_{+}.
\end{eqnarray*}
Hence, we have
\begin{eqnarray*}
s_1(t) & = & s_2(t+1)-|s_2(t)|-[s_3(t)]_{+},\\
s_2(t) & = & \frac{s_{4}(t+1)-s_{3}(t+1)}{2},\\
s_{3}(t) & = & {[s_{3}(t)]_{+}}-x_{6}(t+1),\quad [s_{3}(t)]_{+}=\frac{\max(s_{3}(t+1),s_{4}(t+1))-s_{5}(t+1)}{2},\\
s_{4}(t) & = & x_{7}(t+1)-x_{8}(t+1),\\
s_{5}(t) & = & x_{9}(t+1)-x_{10}(t+1).
\end{eqnarray*}
Therefore,
\begin{eqnarray}
s_1(0)&=&\frac{x_7(3)-x_8(3)-[s_3(2)]_{+}+x_6(3)}{2}-|s_2(0)|-[s_3(0)]_{+},\\
s_2(0)&=&\frac{x_7(2)-x_8(2)-[s_3(1)]_{+}+x_6(2)}{2},\\
s_3(0)&=&[s_3(0)]_{+}-x_6(1),\\
s_4(0)&=&x_7(1)-x_8(1),\\
s_5(0)&=&x_9(1)-x_{10}(1).
\end{eqnarray}

Let $p_t:=[s_3(t)]_+$.
Then we have
\begin{eqnarray*}
&&p_t=\max(D_t,B_t(p_{t+1}))=\left\{
\begin{array}{ll}
D_t & \mbox{if $s_{2}(t)\geq0$},\\
B_t(p_{t+1}) & \mbox{if $s_{2}(t)<0$},
\end{array}
\right.
\end{eqnarray*}
where
\begin{eqnarray*}
D_t&=&\frac{x_7(t+2)-x_8(t+2)-x_9(t+2)+x_{10}(t+2)}{2},\\
B_t(r)&=&\frac{r-x_6(t+2)-x_9(t+2)+x_{10}(t+2)}{2}.
\end{eqnarray*}
Specifically, whenever $s_{2}(t)\geq0$, we have $s_{4}(t+1)\geq s_{3}(t+1)$, and thus
\begin{eqnarray*}
[s_{3}(t)]_{+}=\frac{\max(s_{3}(t+1),s_{4}(t+1))-s_{5}(t+1)}{2}=\frac{s_{4}(t+1)-s_{5}(t+1)}{2}.
\end{eqnarray*}
When $s_{2}(t)<0$, we have
\begin{eqnarray*}
[s_{3}(t)]_{+}=\frac{\max(s_{3}(t+1),s_{4}(t+1))-s_{5}(t+1)}{2}=\frac{s_{3}(t+1)-s_{5}(t+1)}{2}.
\end{eqnarray*}

Moreover, at least one of $s_2(t),s_2(t+1),s_2(t+2)$ is nonnegative. Suppose that $s_2(t)<0$ and $s_2(t+1)<0$.
Since $s_2(t+1)=s_1(t)+|s_2(t)|+p_t<0$,
where $p_t:=[s_3(t)]_+\geq0$, we obtain $s_1(t)<-|s_2(t)|-p_t$.
It follows that
$s_1(t+1)=-s_1(t)+|s_2(t)|+p_t>2|s_2(t)|+2p_t>0$.
Consequently,
$s_2(t+2)=s_1(t+1)+|s_2(t+1)|+p_{t+1}>0$.

Thus, we have
\begin{eqnarray*}
&&p_t=\left\{
\begin{array}{ll}
D_t & \mbox{if $s_{2}(t)\geq0$},\\
B_t(D_{t+1}) & \mbox{if $s_2(t)<0$ and $s_2(t+1)\geq0$},\\
B_t(B_{t+1}(D_{t+2})) & \mbox{if $s_2(t)<0$ and $s_2(t+1)<0$}.
\end{array}
\right.
\end{eqnarray*}
In the last case, $s_2(t+2)>0$. Consequently, \begin{eqnarray*}
p_t=\max(D_t,B_t(D_{t+1}),B_t(B_{t+1}(D_{t+2}))).
\end{eqnarray*}
Thus, $p_0,p_1,p_2$ can be uniquely determined from the values of
$x_6(t),x_7(t),\ldots,x_{10}(t)$, $t\in\llbracket1,6\rrbracket$.
Then substituting $p_0,p_1,p_2$ into the Eqs. (7)--(11), we can uniquely determine $s_1(0),\ldots,s_5(0)$.

Finally, according to
\begin{eqnarray*}
(x_1(0),x_2(0),x_3(0),x_4(0),x_5(0))=S_5^{-1}(s_1(0),s_2(0),s_3(0),s_4(0),s_5(0)),
\end{eqnarray*}
where $S_5^{-1}=-\frac{1}{2}I_5+\frac{1}{6}\mathbf{1}_5\mathbf{1}_5^{\mathsf T}$,
we have
$x_i(0)=-\frac{1}{2}s_i(0)+\frac{1}{6}\sum_{j=1}^{5}s_j(0),~i\in\llbracket1,5\rrbracket$.
Therefore, the entire initial state $\xvec(0)$ can be uniquely recovered from
the output sequence of $y_1(t),\ldots,y_{5}(t)$ over $t\in\llbracket0,6\rrbracket$.
\end{proof}

Next, we show that $\frac{n}{2}$ is a general lower bound for the minimum number of observation nodes.
\begin{theorem}\label{theorem6}
For any general ReLU network with $n$ nodes, if the network is observable, then $m\geq\frac{n}{2}$,
where $m$ denotes the number of observation nodes.
\end{theorem}
\begin{proof}
Suppose that $m<\frac{n}{2}$. Without loss of generality, assume that the observation nodes are $x_{1},x_{2},\ldots,x_{m}$. Consider the following general ReLU network:
\begin{eqnarray*}
x_{i}(t+1) & = & \max(f_{i}(\xvec(t)),0),\quad i\in\llbracket1,n\rrbracket,\\
y_{j}(t) & = & x_{j}(t),\quad j\in\llbracket1,m\rrbracket,
\end{eqnarray*}
where $f_{i}(\xvec(t))=\avec_{i}^{\mathsf T}\xvec(t)+b_{i},\avec_{i}\in\mathbb{R}^{n}$, and $b_{i}\in\mathbb{R}$.

We fix the first $m$ components of the initial state $\xvec(0)$ and denote them by $\uvec=(x_{1}(0),\ldots,x_{m}(0))\in\mathbb{R}^{m}$. Let the remaining $n-m$ components be free and denote them by $\zvec=(x_{m+1}(0),\ldots,x_{n}(0))\in\mathbb{R}^{n-m}$. We show that there exists two distinct initial states $\xvec^{1}(0)=(\uvec,\zvec^{1})$ and $\xvec^{2}(0)=(\uvec,\zvec^{2})$ such that $\xvec(1;\xvec^{1}(0))=\xvec(1;\xvec^{2}(0))$.

For the fixed $\uvec$, each function $f_{i}$ can be written as
\begin{eqnarray*}
f_{i}(\uvec,\zvec)=\alpha_{i}^{\mathsf T}\zvec+\beta_{i}^{\mathsf T}\uvec+b_{i}=\alpha_{i}^{\mathsf T}\zvec+c_{i},
\end{eqnarray*}
where $\avec_{i}=(\beta_{i},\alpha_{i}),\beta_{i}\in\mathbb{R}^{m},\alpha_{i}\in\mathbb{R}^{n-m}$, and $c_{i}=\beta_{i}^{\mathsf T}\uvec+b_{i}\in\mathbb{R}$.

Let $J=\{i\mid \alpha_{i}\neq\textbf{0}_{n-m}\}$. If $J=\emptyset$, then every $f_{i}(\uvec,\zvec)$ is independent of $\zvec$. Hence, for any two distinct $\zvec^{1}$ and $\zvec^{2}$, we have $\xvec(1;\xvec^{1}(0))=\xvec(1;\xvec^{2}(0))$, where $\xvec^{1}(0)=(\uvec,\zvec^{1})$ and $\xvec^{2}(0)=(\uvec,\zvec^{2})$.

Now suppose that $J\neq\emptyset$. For each $i\in J$, the set $H_{i}:=\{\vvec\in\mathbb{R}^{n-m}\mid\alpha_{i}^{\mathsf T}\vvec=0\}$ is a proper hyperplane. Since the finite union of proper hyperplanes cannot cover $\mathbb{R}^{n-m}$, there exist a vector $\vvec\notin\cup_{i\in J}H_{i}$. Hence, $\alpha_{i}^{\mathsf T}\vvec\neq0$ for every $i\in J$.

Since $J$ is finite, we choose a sufficiently large $r>0$ such that, for every $i\in J$,
\begin{eqnarray*}
\operatorname{sgn}(f_{i}(\uvec,r\vvec))=\operatorname{sgn}(\alpha_{i}^{\mathsf T}\vvec),\quad{\rm and}\quad
\operatorname{sgn}(f_{i}(\uvec,-r\vvec))=-\operatorname{sgn}(\alpha_{i}^{\mathsf T}\vvec).
\end{eqnarray*}
In particular, neither $f_{i}(\uvec,r\vvec)$ nor $f_{i}(\uvec,-r\vvec)$ is zero for any $i\in J$. Therefore, for each $i\in J$, exactly one of $f_{i}(\uvec,r\vvec)$ and $f_{i}(\uvec,-r\vvec)$ is positive. It follows that
\begin{eqnarray*}
|\{i\in J\mid f_{i}(\uvec,r\vvec)>0\}|+|\{i\in J\mid f_{i}(\uvec,-r\vvec)>0\}|=|J|\leq n.
\end{eqnarray*}
Consequently, at least one of $r\vvec$ and $-r\vvec$, denoted by $\zvec^{\star}$ satisfies
\begin{eqnarray*}
|\{i\in J\mid f_{i}(\uvec,\zvec^{\star})>0\}|\leq \frac{n}{2}.
\end{eqnarray*}

Define $I:=\{i\in J\mid f_{i}(\uvec,\zvec^{\star})>0\}$. Since $|I|\leq\frac{n}{2}<n-m$, the homogeneous linear system $\alpha_{i}^{\mathsf T}\zvec=0,~i\in I$, has a nonzero solution $\hvec\in\mathbb{R}^{n-m}$.

For every $i\in I$ and every scalar $\varepsilon$, we have
\begin{eqnarray*}
f_{i}(\uvec,\zvec^{\star}+\varepsilon\hvec)=f_{i}(\uvec,\zvec^{\star}),
\end{eqnarray*}
because $\alpha_{i}^{\mathsf T}\hvec=0$.

For every $i\in J\setminus I$, we have $f_{i}(\uvec,\zvec^{\star})<0$. Since $J\setminus I$ is finite and each $f_{i}$ is continuous, there exists a sufficiently small nonzero scalar $\epsilon$ such that $f_{i}(\uvec,\zvec^{\star}+\epsilon\hvec)<0$ for every $i\in J\setminus I$.

For every $i\notin J$, the function $f_{i}(\uvec,\zvec)$ is independent of $\zvec$, and hence
$f_{i}(\uvec,\zvec^{\star}+\epsilon\hvec)=f_{i}(\uvec,\zvec^{\star})$.

Therefore, for every $i\in\llbracket1,n\rrbracket$, we have
\begin{eqnarray*}
\max(f_{i}(\uvec,\zvec^{\star}+\epsilon\hvec),0)=\max(f_{i}(\uvec,\zvec^{\star}),0).
\end{eqnarray*}
Set $\zvec^{1}=\zvec^{\star}$ and $\zvec^{2}=\zvec^{\star}+\epsilon\hvec$. Since $\epsilon\neq0$ and $\hvec\neq\textbf{0}_{n-m}$, we have $\zvec^{1}\neq\zvec^{2}$. Moreover, $\xvec(1;\xvec^{1}(0))=\xvec(1;\xvec^{2}(0))$, where $\xvec^{1}(0)=(\uvec,\zvec^{1})$ and $\xvec^{2}(0)=(\uvec,\zvec^{2})$.

Thus, there exist two distinct initial states that generate the same output sequence. Therefore, the system is not observable whenever $m<\frac{n}{2}$.
\end{proof}
\section{Simulation results}\label{sec:simulation results}
It is generally difficult to directly determine the minimum number of observation nodes for ReLU networks with real-valued states. Therefore, we perform a numerical experiment to verify the lower bound on the minimum number of observation nodes.

We randomly generate 100 ReLU networks with $n$ nodes, where each node has indegree $K$. The initial state is restricted to $\xvec(0)\in[-B,B]^{n}$. For each network, we test observation sets $O\subseteq\{1,2,\ldots,n\}$ with $|O|=m$, which are sampled uniformly at random without replacement. Exhaustive enumeration of all possible observation sets is avoided due to the combinatorial explosion.

For each sampled observation set $O$, we search for a counterexample, namely two distinct initial states $\xvec^{1}(0),\xvec^{2}(0)$ that generate two output sequences $y_{O}^{1}(t)$ and $y_{O}^{2}(t)$, where $y_{O}(t)=x_{O}(t)$, satisfying $\max_{t\leq T}\|y_{O}^{1}(t)-y_{O}^{2}(t)\|_{\infty}\leq \delta$. If such a counterexample is found, the observation set $O$ is regarded as non-observable.

In our experiments with $(n,K)=(10,5),(12,6),(14,7),B=50,\delta=10^{-2}$, and $T=40$, we consistently find counterexamples for all tested observation sets with $|O|<\frac{n}{2}$ across all 100 randomly generated networks (see Fig. \ref{fig}). This indicates that observation sets of size less than $\frac{n}{2}$ are empirically insufficient, thereby providing strong empirical evidence supporting the lower bound $m\geq \frac{n}{2}$.
\begin{figure}[t]
  \centering
  \includegraphics[width=0.45\textwidth]{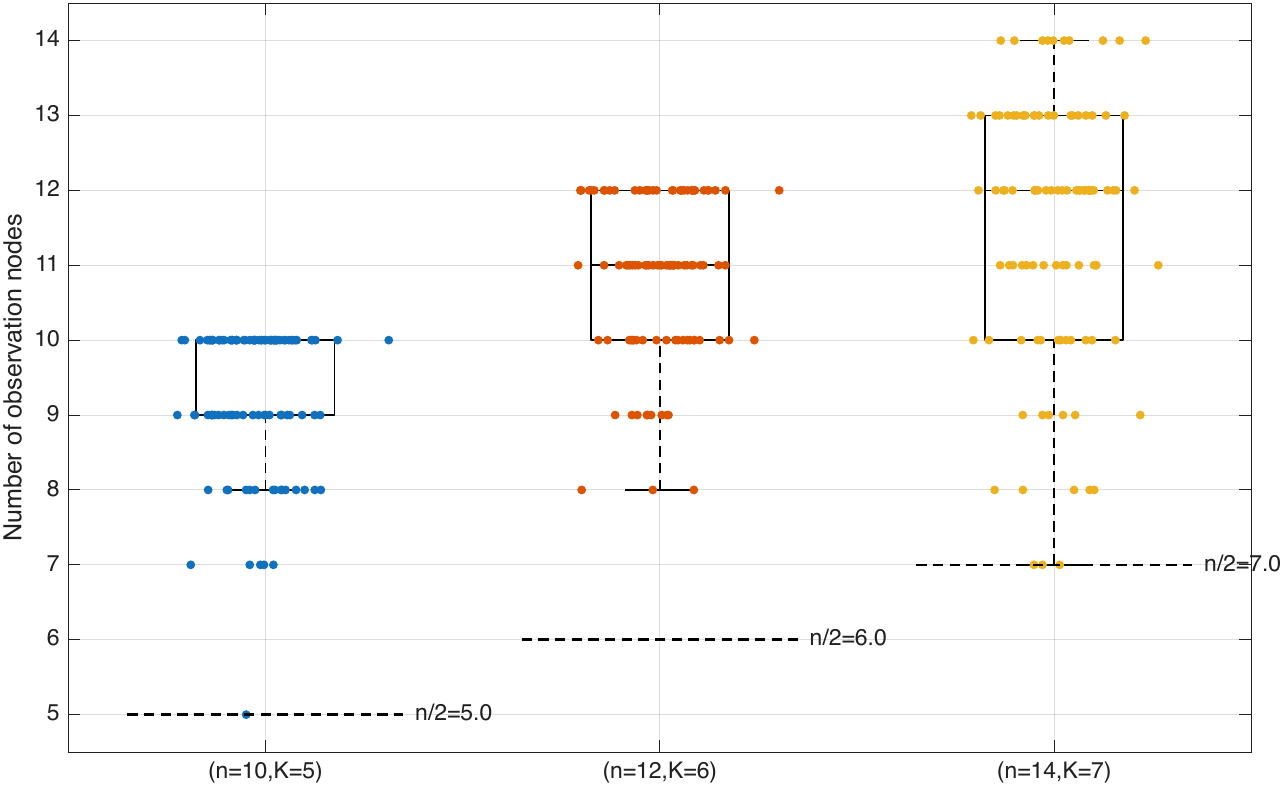}
  \caption{Empirical verification of the lower bound on the number of observation nodes for ReLU networks with different network sizes $n$ and indegrees $K$. In each case, all tested observation sets with $|O|<n/2$ admit counterexamples across 100 randomly generated networks, indicating that the minimum number of observation nodes empirically satisfies $m\geq n/2$.}
  \label{fig}
\end{figure}
\section{Conclusion}\label{sec:conclusion}
This paper investigated minimum observation node in LT and ReLU networks from two complementary perspectives: the effect of the node update rule over a common binary state domain and the effect of the admissible state domain within the ReLU framework.

For binary-state dynamics, we constructed a class of $K$-LT networks whose initial states can be uniquely determined from the finite output sequence of a single observation node. Since at least one observation node is necessary, this establishes that the best-case minimum for the considered $K$-LT class is exactly one. We also established a dynamical equivalence between binary $K$-ReLU networks and $K$-AND BNs, allowing the corresponding $K$-dependent observation-node bounds to be transferred to binary $K$-ReLU networks. This comparison shows that, even over the same binary state domain, the update rule can fundamentally alter the minimum number of observation nodes and determine whether information from unmeasured states can be encoded into the temporal outputs of a small set of measured nodes.

For positive ReLU networks, where all relevant ReLU units remain active, the dynamics reduce to a linear system and the finite-horizon observability problem is governed by the classical Kalman observability condition. General real-valued states lead to a fundamentally different situation because the activation pattern may vary with the initial condition. On any nonempty open region associated with a fixed activation sequence, the finite-horizon observation map is affine. We proved that this map cannot be injective if the rank of its linear part is smaller than the dimension of the freely varying initial-state components. This affine-region rank obstruction yields a general lower bound of $\frac{n}{2}$ observation nodes, without imposing any restriction on the number of state variables on which each node update depends. When $n=2K$, we constructed a $K$-ReLU network and provided an explicit finite-step procedure for reconstructing its entire initial state from measurements of exactly $K$ nodes. The construction attains the general lower bound, proving that the best-case minimum for the considered class of $K$-ReLU networks is exactly $K$.

These results demonstrate that extremal observation requirements are determined not only by network size, but also by the node update rule, the admissible state domain, and the resulting activation patterns. In finite-state threshold dynamics, information distributed across different state variables can be encoded into the temporal sequence of a single measured node. In general real-valued ReLU dynamics, by contrast, state-dependent deactivation can eliminate observable directions and leave distinct initial states indistinguishable. The rank obstruction provides a certificate of nonobservability for insufficient sensing configurations, whereas the explicit constructions characterize the best observation performance attainable within the corresponding network classes.

\end{document}